\documentclass{llncs}

\usepackage[utf8]{inputenc}
\usepackage{amsfonts}
\usepackage{amssymb,amsmath}
\usepackage{gastex}
\usepackage{hyperref}
\usepackage{tikz}
\usepackage{nicefrac}

\newcommand{\A}{\mathcal{A}}
\newcommand{\B}{\mathcal{B}}
\newcommand{\C}{\mathcal{C}}
\newcommand{\D}{\mathcal{D}}

\newcommand{\M}{\mathcal{C}}
\newcommand{\set}[1]{\{ #1 \}}
\newcommand{\val}[1]{\text{val}(#1)}
\newcommand{\prob}[1]{\mathbb{P}_{#1}}

\newcommand{\merge}{\textrm{next\_transition}}
\newcommand{\chck}{\textrm{check}}
\newcommand{\apply}{\textrm{apply}}
\newcommand{\finish}{\textrm{next\_word}}
\newcommand{\wait}{\textrm{wait}}
\newcommand{\supp}{\textrm{supp}}

\begin{document}

\title{Two Recursively Inseparable Problems\\
for Probabilistic Automata\thanks{The research leading to these results has received funding from the European Union's Seventh Framework Programme (FP7/2007-2013) under grant agreement n\textsuperscript{o} 259454 (GALE) and from the French Agence Nationale de la Recherche projects EQINOCS (ANR-11-BS02-004) and STOCH-MC (ANR-13-BS02-0011-01).}}

\author{Nathana\"el Fijalkow\inst{1,2} \and Hugo Gimbert\inst{3} \and Florian Horn\inst{1} \and Youssouf Oualhadj\inst{4}}
\institute{LIAFA, Université Paris 7, France,\\
University of Warsaw, Poland,\\
LaBRI, Université de Bordeaux, France,\\
Universit\'e de Mons, Belgium.}

\maketitle

\begin{abstract}
This paper introduces and investigates decision problems for \textit{numberless} probabilistic automata,
\textit{i.e.} probabilistic automata where the \textit{support} of each probabilistic transitions is specified, but the exact values of the probabilities are \textit{not}.
A numberless probabilistic automaton can be \emph{instantiated} into a probabilistic automaton
by specifying the exact values of the non-zero probabilistic transitions.

We show that the two following properties of numberless probabilistic automata are
recursively inseparable:
\begin{itemize}
	\item[\textbullet] all instances of the numberless automaton have value $1$,
	\item[\textbullet] no instance of the numberless automaton has value $1$.
\end{itemize}
\end{abstract}

\section{Introduction}
In 1963 Rabin~\cite{Rabin63} introduced the notion of probabilistic
automata,
which are finite automata able to randomise over transitions.
A probabilistic automaton has a finite set of control states $Q$, and
processes finite words;
each transition consists in updating the control state
according to a given probabilistic distribution determined by the
current state and the input letter.
This powerful model has been widely studied and has applications in many
fields
like software verification~\cite{CDHR07}, image processing~\cite{CK97},
computational biology~\cite{DEKM99} and speech processing~\cite{M97}.

Several algorithmic properties of probabilistic automata have been
considered in the literature,
sometimes leading to efficient algorithms.
For instance, \textit{functional equivalence} is decidable in polynomial
time~\cite{Schutzenberger61,Tzeng92},
and even faster with randomised algorithms, which led to applications in
software verification~\cite{KMOWW11}.

However, many natural decision problems are undecidable,
and part of the literature on probabilistic automata is about
\textit{intractability results}.
For example the \textit{emptiness}, the \textit{threshold isolation} and the \textit{value 1}
problems are undecidable~\cite{Paz71,Bertoni77,GO10}.

A striking result due to Condon and Lipton~\cite{CL89} states that, for every $\epsilon>0$, the two following problems are
recursively inseparable: given a probabilistic automaton $\A$,
\begin{itemize}
	\item[\textbullet] does $\A$ accept some word with probability greater than $1 - \varepsilon$?
	\item[\textbullet] does $\A$ accept every word with probability less than $\varepsilon$?
\end{itemize}

\vskip1em
In the present paper we focus on \emph{numberless} probabilistic automata,
\textit{i.e.} probabilistic automata whose non-zero probabilistic transitions are specified
but the exact values of the probabilities are not.
A numberless probabilistic automaton can be \emph{instantiated} into a probabilistic automaton
by specifying the exact values of the non-zero probabilistic transitions (see Section~\ref{sec:defs} for formal definitions).

\vskip1em
The notion of numberless probabilistic automaton is motivated by the following example.
Assume we are given a digital chip modelled
as a finite state machine controlled by external inputs.
The internal transition structure of the chip is known but the transitions themselves are not observable.
We want to compute an initialisation input sequence that puts the chip in a particular initial state.
In case some of the chip components have failure probabilities,
this can be reformulated as a value $1$ problem for the underlying probabilistic automaton:
is there an input sequence whose acceptance probability is arbitrarily close to $1$?
Assume now that the failure probabilities are not fixed \emph{a priori}
but we can tune the quality of our components
and choose the failure probabilities,
for instance by investing in better components.
Then we are dealing with a numberless probabilistic automaton
and we would like to determine whether it can be instantiated into a
probabilistic automaton with value $1$,
in other words we want to solve an \emph{existential value $1$ problem}
for the numberless probabilistic automaton.
If the failure probabilities are unknown then we are facing a second kind
of problem called the \emph{universal value $1$ problem}:
determine whether all instances of the automaton
have value $1$.
We also consider variants where the freedom to choose transition probabilities
is restricted to some intervals, that we call \emph{noisy value 1 problems}.


\vskip1em
One may think that relaxing the constraints on the exact transition
probabilities makes things algorithmically much easier.
However this is not the case, and we prove that the existential and universal value $1$ problems are recursively inseparable:
given a numberless probabilistic automaton $\M$,
\begin{itemize}
	\item[\textbullet] do all instances of $\M$ have value $1$?
	\item[\textbullet] does no instance of $\M$ have value $1$?
\end{itemize}

This result is actually a corollary of a generic construction which
constitutes the technical core of the paper
and has the following properties.
For every simple probabilistic automaton $\A$,
we construct a numberless probabilistic automaton $\M$ such that
the three following properties are equivalent:
\begin{enumerate}
	\item[(i)] $\A$ has value $1$,
	\item[(ii)] one of the instances of $\M$ has value $1$,
	\item[(iii)] all instances of $\M$ have value $1$.
\end{enumerate}

The definitions are given in Section~\ref{sec:defs}.
The main technical result appears in Section~\ref{sec:numberless_undecidability},
where we give the generic construction whose properties are described above.
In Section~\ref{sec:consequences}, we discuss the implications of our results,
first for the noisy value $1$ problems,
and second for probabilistic B\"uchi automata~\cite{BBG12}.


\section{Definitions}
\label{sec:defs}
Let $A$ be a finite alphabet. A (finite) word $u$ is a (possibly empty) sequence of letters $u = a_0 a_1 \cdots a_{n-1}$;
the set of finite words is denoted by $A^*$. 

A probability distribution over a finite set $Q$ is a function $\delta : Q \rightarrow \mathbb{Q}_{\ge 0}$
such that $\sum_{q \in Q} \delta(q) = 1$; we denote by $\frac{1}{3} \cdot q + \frac{2}{3} \cdot q'$ the distribution
that picks $q$ with probability $\frac{1}{3}$ and $q'$ with probability $\frac{2}{3}$,
and by $q$ the trivial distribution picking $q$ with probability~$1$. The support of a distribution $\delta$ is the set of states picked with positive probability,
\textit{i.e.}, $\supp(\delta) = \set{q \in Q \mid \delta(q) > 0}$.
Finally, the set of probability distributions over $Q$ is $\D(Q)$.

\begin{definition}[Probabilistic automaton]
A \emph{probabilistic automaton} (PA) is a tuple $\A = (Q, A, q_0, \Delta, F)$,
where $Q$ is a finite set of states, $A$ is the finite input alphabet, 
$q_0 \in Q$ is the initial state, $\Delta: Q \times A \rightarrow D(Q)$ is the \emph{probabilistic transition function}, and $F \subseteq Q$ is the set of accepting states.
\end{definition}


For convenience, we also use $\prob{\A}(s \xrightarrow{u} t)$ to denote the probability of going from the state $s$ to the state $t$ reading $u$, $\prob{\A}(s \xrightarrow{u} S)$ to denote the probability of going from the state $s$ to a state in $S$ reading $u$, and $\prob{\A}(u)$, the \emph{acceptance probability} of a word $u \in A^*$, to denote $\prob{\A}(q_0 \xrightarrow{u} F)$.

We often consider the case of \emph{simple probabilistic automata}, where the transition probabilities can only be $0$, $\nicefrac{1}{2}$, or $1$.

\begin{definition}[Value]
The \emph{value} of a PA $\A$, denoted $\val{\A}$,
is the supremum acceptance probability over all input words
\[
\val{\A} = \sup_{u \in A^*} \prob{\A}(u)\enspace.
\]
\end{definition}

\noindent The value $1$ problem asks,
given a (simple) PA $\A$ as input, 
whether $\val{\A} = 1$.

\begin{theorem}[\cite{GO10}]\label{thm:value_1}
The value $1$ problem is undecidable for simple PA.
\end{theorem}

\begin{definition}[Numberless probabilistic automaton]
A \emph{numberless probabilistic automaton} (NPA) is a tuple $\A = (Q, A, q_0, T, F)$,
where $Q$ is a finite set of states, $A$ is the finite input alphabet, 
$q_0 \in Q$ is the initial state, $T \subseteq Q \times A \times Q$ is the \emph{numberless transition function},
and $F \subseteq Q$ is the set of accepting states.
\end{definition}

The numberless transition function $T$ is an abstraction of probabilistic transition functions. We say that $\Delta$ is consistent with $T$ if for all letters $a$ and states $s$ and $t$, $\Delta(s,a,t) > 0$ if, and only if $(s,a,t) \in T$.

A numberless probabilistic automaton is an equivalence class of probabilistic automata, which share the same set of states, input alphabet, initial and accepting states, and whose transition functions have the same support.

A NPA $\A = (Q, A, q_0, T, F)$ together with a probabilistic transition function $\Delta$ consistent with $T$ defines a PA $\A[\Delta] = (Q, A, q_0, \Delta, F)$. Conversely, a PA $\A = (Q, A, q_0, \Delta, F)$ induces an underlying NPA $[\A] = (Q, A, q_0, T, F)$, where $T \subseteq Q \times A \times Q$ is defined by $(q,a,p) \in T$ if $\Delta(q,a)(p) > 0$.

%

%

We consider two decision problems for NPA:
\begin{itemize}
	\item \textbf{The existential value $1$ problem:} given a NPA $\A$, determine whether
there exists $\Delta$ such that $\val{\A[\Delta]} = 1$.
	\item \textbf{The universal value $1$ problem:} given a NPA $\A$, determine whether
for all $\Delta$, we have $\val{\A[\Delta]} = 1$.
\end{itemize}

\begin{proposition}
There exists a NPA such that:
\begin{itemize}
	\item there exists $\Delta$ such that $\val{\A[\Delta]} = 1$,
	\item there exists $\Delta'$ such that $\val{\A[\Delta']} < 1$.
\end{itemize}
\end{proposition}

\begin{figure}[!ht]
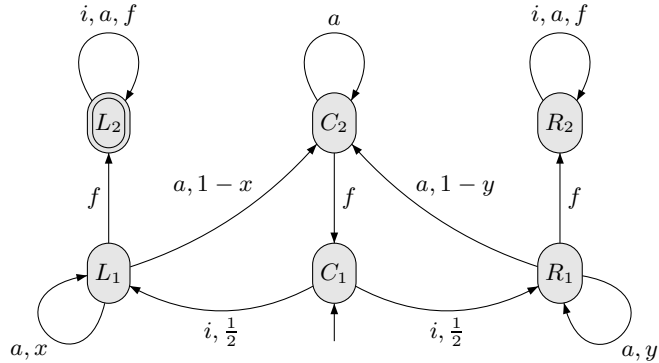

\begin{center}
\begin{gpicture}(60,40)(0,0)
	\gasset{Nadjust=w,fillgray=.9}

  	\node[Nmarks=i,iangle=-90](C1)(30,10){$C_1$}
  	\node(C2)(30,30){$C_2$}
  	\node(L1)(0,10){$L_1$}
  	\node[Nmarks=r](L2)(0,30){$L_2$}
  	\node(R1)(60,10){$R_1$}
  	\node(R2)(60,30){$R_2$}

	\drawloop(C2){$a$}

	\drawedge(C2,C1){$f$}
  	\drawedge[curvedepth=5,ELside=l](C1,L1){$i,\frac{1}{2}$}
  	\drawedge[curvedepth=-3,ELside=l](L1,C2){$a,1-x$}
  	\drawedge(L1,L2){$f$}
	\drawloop[loopangle=-135](L1){$a,x$}
	\drawloop[loopangle=90](L2){$i,a,f$}

  	\drawedge[curvedepth=-5,ELside=r](C1,R1){$i,\frac{1}{2}$}
  	\drawedge[curvedepth=3,ELside=r](R1,C2){$a,1-y$}
  	\drawedge[ELside=r](R1,R2){$f$}
	\drawloop[loopangle=-45](R1){$a,y$}
	\drawloop(R2){$i,a,f$}
\end{gpicture}
\caption{This NPA has value $1$ if and only if $x > y$}
\label{fig:variation}
\end{center}

\end{figure}

In this automaton, adapted from~\cite{GO10,FGO12}, the shortest word that can be accepted is $i \cdot f$, as $i$ goes from $C_1$ to $L_1$, and $f$ goes from $L_1$ to $L_2$. However, there are as much chances to go to $R_2$, so the value of $i\cdot f$ is $\nicefrac{1}{2}$.

If $x$ is strictly less than $y$, one can tip the scales to the left by adding $a$'s between the $i$ and $f$ : each time, the run will have more chances to stay left than to stay right. After reading $i \cdot a^n \cdot f$, the probability of reaching $L_2$ is equal to $x^n$, while the probability of reaching $R_2$ is only $y^n$. There is also a very high chance that the run went back to $C_1$, but from there we can simply repeat our word an arbitrary number of times.

Let $x$, $y$, and $\varepsilon$ be three real numbers such that $0 \le y < x \le 1$, and $0 < \varepsilon \le 1$. There is an integer $n$ such that $\nicefrac{x^n}{(x^n+y^n)}$ is greater than $1-\nicefrac{\varepsilon}{2}$, and an integer $m$ such that $\left(1-x^n-y^n\right)^m$ is less than $\nicefrac{\varepsilon}{2}$. The word $\left(i\cdot a^n \cdot f\right)^m$ is accepted with probability greater than $1-\varepsilon$.

On the other hand, if $x \le y$, there is no word with value higher than $\nicefrac{1}{2}$.


\section{Recursive inseparability for numberless value $1$ problems}
\label{sec:numberless_undecidability}
In this section, we prove the following theorem:

\begin{theorem}\label{thm:main}
The two following problems for numberless probabilistic automata
are recursively inseparable:
\begin{itemize}
	\item all instances have value $1$,
	\item no instance has value $1$.
\end{itemize}
\end{theorem}

Recall that two decision problems $A$ and $B$
are recursively inseparable if their languages $L_A$ and $L_B$ of accepted inputs
are disjoint and there exists no recursive language $L$
such that $L_A\subseteq L$ and $L \cap L_B = \emptyset$.

Note that it implies that both $A$ and $B$ are undecidable.

Equivalently, this means there exists no terminating algorithm which has the following behaviour on input $x$:
\begin{itemize}
	\item if $x \in L_A$, then the algorithm answers ``YES''
	\item if $x \in L_B$, then the algorithm answers ``NO''.
\end{itemize}
On an input that belongs neither to $L_A$ nor to $L_B$, the algorithm's answer can be either ``YES'' or ``NO''. 

%
%
%
%
%
%
%

\subsection{Overall construction}

\begin{lemma}\label{lem:construction}
There exists an effective construction 
which takes as input a simple PA $\A$ 
and constructs a NPA $\M$ such that
\[
\val{\A} = 1\ \Longleftrightarrow\ \forall \Delta, \val{\M[\Delta]} = 1\ \Longleftrightarrow\ \exists \Delta, \val{\M[\Delta]} = 1\ .\]
\end{lemma}

We first explain how Lemma~\ref{lem:construction} implies Theorem~\ref{thm:main}.
Assume towards contradiction that the problems ``all instances have value $1$'' and ``no instance has value $1$''
are recursively separable.
Then there exists an algorithm $A$ taking a NPA as input and such that:
if all instances have value $1$, then it answers ``YES'', and if no instance has value $1$, then it answers ``NO''.
We show using Lemma~\ref{lem:construction} 
that this would imply that the value $1$ problem is decidable for simple PA, contradicting Theorem~\ref{thm:value_1}.
Indeed, let $\A$ be a simple PA, applying the construction yields a NPA $\M$ such that
\[
\val{\A} = 1\ \Longleftrightarrow\ \forall \Delta, \val{\M[\Delta]} = 1\ \Longleftrightarrow\ \exists \Delta, \val{\M[\Delta]} = 1\ .\]
In particular, either all instances of $\M$ have value $1$, or no instance of $\M$ has value $1$.
Hence, if it answers ``YES'' then $\val{\A} = 1$ and it if answers ``NO'' then $\val{\A} < 1$,
allowing to decide whether $\A$ has value $1$ or not.
This concludes the proof of Theorem~\ref{thm:main} assuming Lemma~\ref{lem:construction}.

\vskip1em

The construction follows two steps.

The first step is to build from $\A$ a family of PA's $\B_\lambda$ whose transitions are all of the form $q \xrightarrow{a} (\lambda \cdot r, (1-\lambda) \cdot s)$ in such a way that $\B_\lambda$ has the same value as $\A$ \emph{for any value of $\lambda$}. Note that, while all the $\B_\lambda$'s belong to the same NPA, they are \emph{not}, in general, a NPA: for example, if $\A$ were the simple version of the  automaton of Figure~\ref{fig:variation}, the $B_\lambda$'s would be the instances where $x = y = \lambda$, while the underlying NPA would also include the cases where $x \neq y$.

The second step is to build from the $\B_\lambda$'s a NPA $\M$ such that, for each probabilistic transition function $\Delta$, there is a $\lambda$ such that $\M[\Delta]$ has value 1 if, and only if $\B_\lambda$ has value 1.

It follows that: $$\begin{array}{rcl}
					\exists \Delta, \val{\M[\Delta]} = 1	&	\Longrightarrow	&	\exists \lambda, \val{\B_\lambda} = 1\\
										&	\Longrightarrow	&	\val{\A} = 1\\
										&	\Longrightarrow &	\forall \lambda, \val{\B_\lambda} = 1\\
										&	\Longrightarrow &	\forall \Delta, \val{\M[\Delta]} = 1\ .
		\end{array}$$

%
%
%
%


\subsection{The fair coin construction}
\label{subsec:fair_coin}
Let $\A = (Q, A, q_0, \Delta, F)$ be a simple PA over the alphabet $A$. We construct a family of PAs $(\B_\lambda)_{\lambda\in]0,1[}$ over the alphabet $B = A\cup\{\sharp\}$, whose transitions have probabilities $0$, $\lambda$, $1 - \lambda$ or $1$, as follows.

The automaton $\B_\lambda$ is a copy of $\A$ where each transition of $\A$ is replaced by the gadget illustrated in Figure~\ref{fig:fair_coin} (for simplicity, we assume that all the transitions of $\A$ are probabilistic). The initial and final states are the same in $\A$ and $\B_\lambda$.

The left hand side shows part of automaton $\A$: a probabilistic transition from $q$ reading $a$, leading to $r$ or $s$ each with probability half.
The right hand side shows how this behaviour is simulated by $\B_\lambda$: the letter $a$ leads to an intermediate state $q_a$, from which we can read a new letter $\sharp$. Each time a pair of $\sharp$'s is read, the automaton $B_\lambda$ goes to $r$ with probability $\lambda \cdot (1-\lambda)$, goes to $s$ with probability $(1-\lambda) \cdot \lambda$, and stays in $q_a$ with probability $\lambda^2 + (1-\lambda)^2$. Reading a letter other than $\sharp$ while the automata is still in one of the new states leads to the new sink state $\bot$, which is not accepting.
Thus, the probability of going to $r$ is equal to the probability of going to $s$, and we can make the probability of a simulation failure as low as we want by increasing the number of $\sharp$'s between two successive ``real'' letters.

%

\begin{figure}
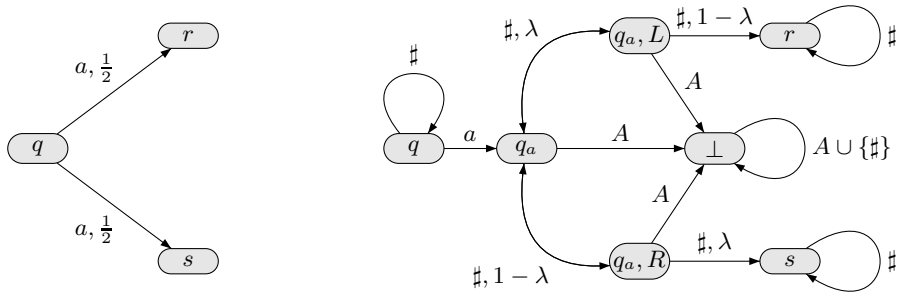

\begin{center}
\begin{gpicture}(110,33)(0,0)
	\gasset{Nadjust=h,fillgray=0.9}

  	\node(q)(0,15){$q$}
  	\node(q0)(20,30){$r$}
  	\node(q1)(20,0){$s$}

  	\drawedge(q,q0){$a,\frac{1}{2}$}
  	\drawedge[ELside=r](q,q1){$a,\frac{1}{2}$}

  	\node(qr)(50,15){$q$}
  	\node(qar)(65,15){$q_a$}
  	\node(qaL)(80,30){$q_a,L$}
  	\node(qaR)(80,0){$q_a,R$}
  	\node(q0r)(100,30){$r$}
  	\node(q1r)(100,0){$s$}
  	\node(bot)(90,15){$\bot$}

  	\drawloop[loopangle=90](qr){$\sharp$}
	\drawedge(qr,qar){$a$}
	\drawedge[curvedepth=8](qar,qaL){$\sharp,\lambda$}
	\drawedge[curvedepth=-8](qaL,qar){}
	\drawedge[curvedepth=-8,ELside=r](qar,qaR){$\sharp,1 - \lambda$}
	\drawedge[ELside=r,curvedepth=8](qaR,qar){}
  	\drawedge(qaL,q0r){$\sharp,1 - \lambda$}
  	\drawedge(qaR,q1r){$\sharp,\lambda$}
  	\drawloop[loopangle=0](q0r){$\sharp$}
  	\drawloop[loopangle=0](q1r){$\sharp$}
  	
  	\drawedge(qar,bot){$A$}
  	\drawedge(qaL,bot){$A$}
  	\drawedge(qaR,bot){$A$}
	\drawloop[loopangle=0](bot){$A\cup\{\sharp\}$}

\end{gpicture}
\end{center}
\caption{The fair coin gadget.}
\label{fig:fair_coin}
\end{figure}

	Let $u$ be a word of $A^*$. We denote by $[u]^k$ the word of $B^*$ where each letter $a \in A$ of $u$ is replaced by $a\cdot \sharp^{2k}$. Conversely, if $w$ is a word of $B^*$, we denote by $\tilde{w}$ the word obtained from $w$ by removing all occurrences of the letter $\sharp$.

	Intuitively, a run of $\A$ on the word $u$ is simulated by a run of $\B_\lambda$ on the word $[u]^k$. Whenever there is a transition in $\A$, $\B_\lambda$ makes $k$ attempts to simulate it through the gadget of Figure~\ref{fig:fair_coin}, and each attempt succeeds with probability $\left(1 - 2 \lambda \cdot (1 - \lambda) \right)$, so each transition fails with probability:
\[
A_{\lambda,k} = 1 - \left(1 - 2 \lambda \cdot (1 - \lambda) \right)^k\ .
\]

%
%
 
\begin{proposition}
\label{prop:fair_coin}
The probabilistic automaton $\B_\lambda$ satisfies:
\begin{enumerate}
	\item For all $q,r \in Q$, $a \in A$, and $k \in \mathbb{N}$, $\prob{\B_\lambda}(q \xrightarrow{[a]^k} r) = A_{\lambda,k} \cdot \prob{\A}(q \xrightarrow{a} r)\ ,$
	\item For all $q,r \in Q$, $u \in A^*$, and $k \in \mathbb{N}$, $\prob{\B_\lambda}(q \xrightarrow{[u]^k} r) = A_{\lambda,k}^{|u|} \cdot \prob{\A}(q \xrightarrow{u} r)\ .$
	\item For all $q,r \in Q$, $w \in (A\cup\{\sharp\})^*$, and $k \in \mathbb{N}$, $\prob{\B_\lambda}(q \xrightarrow{w} r) \le \prob{\A}(q \xrightarrow{\tilde{w}} r)\ .$
\end{enumerate}
\end{proposition}

It follows from Proposition~\ref{prop:fair_coin} that for any $\lambda$, the value of $\B_\lambda$ is equal to the value of $\A$.


\subsection{The simulation construction}
\label{subsec:simulation}
All the $\B_\lambda$'s induce the same NPA, that we denote by $\B$. The problem is that there are many other instances of $\B$, whose values may be higher than the value of $\A$ (recall the example of Figure~\ref{fig:variation}, where the $B_\lambda$ have value $\nicefrac{1}{2}$, while there are instances of $\B$ with value $1$). In this subsection, we construct a NPA $\C$ (over an extended alphabet $C$) whose instances simulate all the $\B_\lambda$'s, but only them.

The idea is that the new NPA should only have \emph{one} probabilistic transition. An instance of this transition translates to a value for $\lambda$. Figure~\ref{fig:naive} describes a first attempt at this (notice that our convention is that an non-drawn transition means a loop, rather than a transition to a sink state).


\begin{figure}
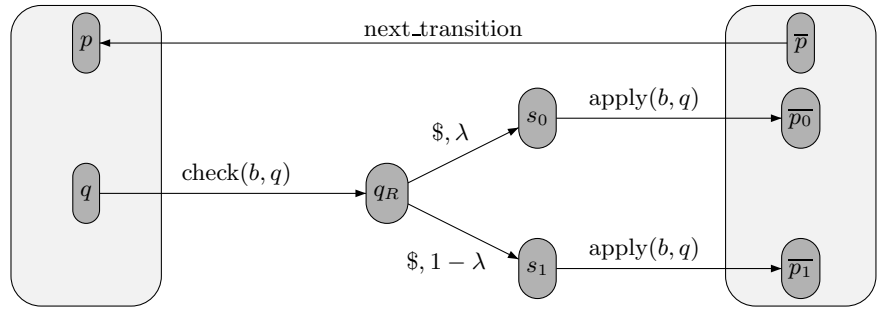

\begin{center}
\begin{gpicture}(115,40)(-10,45)
  	\node[Nw=20,Nh=40,fillgray=0.95](gauche)(0,65){}

  	\node[Nw=20,Nh=40,fillgray=0.95](droite)(95,65){}

%
	\gasset{Nadjust=w,fillgray=0.7}

  	\node(p)(0,80){$p$}
  	\node(q)(0,60){$q$}

  	\node(triangle)(40,60){$q_R$}

  	\node(s0)(60,70){$s_0$}
  	\node(s1)(60,50){$s_1$}

  	\node(op)(95,80){$\overline{p}$}
  	\node(p0)(95,70){$\overline{p_0}$}
  	\node(p1)(95,50){$\overline{p_1}$}


  	\drawedge[ELside=r](op,p){$\merge$}

  	\drawedge(q,triangle){$\chck(b,q)$}

  	
  	\drawedge(triangle,s0){$\$,\lambda$}
  	\drawedge[ELside=r](triangle,s1){$\$,1-\lambda$}


  	\drawedge(s0,p0){$\apply(b,q)\qquad $}
  	\drawedge(s1,p1){$\apply(b,q)\qquad $}

\end{gpicture}
\end{center}
\caption{Naive fusion of the probabilistic transitions.}
\label{fig:naive}
\end{figure}

In this automaton, that we call $\B'$, there are two copies of the set of states, and a single shared probabilistic transition $q_R \xrightarrow{\$} (\lambda \cdot s_0, (1-\lambda) \cdot s_1)$ between them. In order to make all the probabilistic transitions of $\B$ happen in this center area, we use new letters to detect where the runs come from before the probabilistic transition, and where it should go afterwards.

For each pair of a letter $b$ in $B$ and a state $q$ in $Q$, we introduce two new letters $\chck(b,q)$ and $\apply(b,q)$. The letter $\chck(b,q)$ loops over each state except the left copy of $q$, which goes to $q_R$. The letter $\apply(b,q)$ loops over each state except $s_0$, from where it goes to the $\lambda$-valued successor of $(q,b)$ in $\B_\lambda$, and $s_1$, from where it goes to the $(1-\lambda)$-valued successor of $(q,b)$ in $\B_\lambda$. The new letter $\merge$ sends the run back to the left part once each possible state has been tested.

Thus, if we define the morphism $\widehat{\_}$ by its action on letters:
\[
\widehat{b} = \chck(b,q_0) \cdot \$ \cdot \apply(b,q_0)\ \cdots\ \chck(b,q_{n-1}) \cdot \$ \cdot \apply(b,q_{n-1}) \cdot \merge\ ,
\]
where the $q_i$'s are the states of $\B_\lambda$, we get for any word $u$ on the alphabet $B$, $\prob{\B_\lambda}(u) = \prob{\B'_\lambda}(\widehat{u})$.

The problem with this automaton is that one can ``cheat'', either by not testing an unwelcome state, or by changing state and letter between a \chck\ and the subsequent \apply. In order to avoid this kind of behaviour, we change the automaton so that it is necessary to win an arbitrarily large number of successive times in order to approach the value 1, and we can test whether the word is fair after the first successful attempt. A side effect is that the simulation only works for the value 1: for other values, it might be better to cheat. The resulting automaton is described in Figure~\ref{fig:simulation}.

\begin{figure}
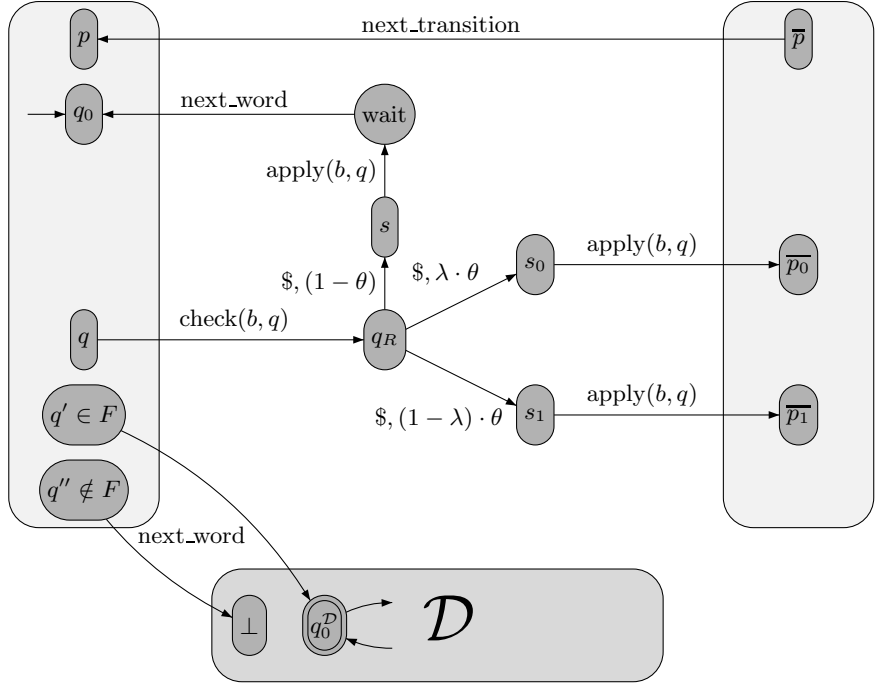

\begin{center}
\begin{gpicture}(90,80)(0,0)
  	\node[Nw=20,Nh=70,fillgray=0.95](gauche)(0,50){}

  	\node[Nw=20,Nh=70,fillgray=0.95](droite)(95,50){}

  	\node[Nw=60,Nh=15,fillgray=0.85](bas)(47,2){}
	\put(45,0){\begin{Huge}$\D$\end{Huge}}

	\gasset{Nadjust=w,fillgray=0.7}

  	\node(p)(0,80){$p$}
  	\node[Nmarks=i,iangle=180](q_0)(0,70){$q_0$}
  	\node(q)(0,40){$q$}
  	\node(q')(0,30){$q' \in F$}
  	\node(q'')(0,20){$q'' \notin F$}

  	\node(w)(40,70){$\wait$}
  	\node(s)(40,55){$s$}
  	\node(triangle)(40,40){$q_R$}

  	\node(s0)(60,50){$s_0$}
  	\node(s1)(60,30){$s_1$}

  	\node(op)(95,80){$\overline{p}$}
  	\node(p0)(95,50){$\overline{p_0}$}
  	\node(p1)(95,30){$\overline{p_1}$}

	\node(bot)(22,2){$\bot$}
  	\node[Nmarks=r](r0)(32,2){$q_0^\D$}

  	\node[Nframe=n,fillgray=.85](toto)(42,5){}
	\node[Nframe=n,fillgray=.85](tata)(42,-1){}
	\drawedge[curvedepth=1](r0,toto){}
	\drawedge[curvedepth=1](tata,r0){}

  	\drawedge[ELside=r](op,p){$\merge$}

  	\drawedge(q,triangle){$\chck(b,q)$}

  	\drawedge[ELside=r](w,q_0){$\finish$}
  	\drawedge[curvedepth=4,ELside=r](q',r0){$\finish$}
	\drawedge[curvedepth=-2](q'',bot){}
  	
  	\drawedge(triangle,s0){$\$, {\lambda}\cdot{\theta}$}
  	\drawedge[ELside=r](triangle,s1){${\$,(1-\lambda)}\cdot{\theta}$}
  	\drawedge(triangle,s){$\$,(1-\theta)$}


  	\drawedge(s0,p0){$\apply(b,q)\qquad $}
  	\drawedge(s1,p1){$\apply(b,q)\qquad $}

  	\drawedge[ELside=l](s,w){$\apply(b,q)$}
\end{gpicture}
\end{center}
\caption{The Numberless Probabilistic Automaton $\C$.}
\label{fig:simulation}
\end{figure}

The structure of the automaton of Figure~\ref{fig:naive} is still there, but it has been augmented with an extra layer of scrutiny: each time we use the probabilistic transition, there is now a positive probability $(1-\theta)$ to go to a new wait state. There is also a new letter $\finish$ which has the following effect:
\begin{itemize}
	\item if the run is in an accepting state, it goes to the initial state of the fairness checker $\D$;
	\item if the run is in a non-accepting state, it goes to the non-accepting sink state $\bot$ of $\D$; 
	\item if the run is in the wait state, it goes back to the initial state.
\end{itemize}

The fairness checker $\D$ is a deterministic automaton which accepts the language $\set{\widehat{u} \cdot \finish \mid u \in B^*}^*$. Its final state is the only final state in all of $\C$.

Intuitively, a run can still cheat before the \emph{first} \finish\ letter, but the benefits of doing so are limited: the probability that $\C[\lambda,\theta]$ accepts a word at that point is at most $\theta$ (except if the empty word is accepting, but that case is trivial). After that point, cheating is risky: if the run already reached $\D$, a false move will send it to the sink.

A more formal proof follows. A simple inspection of the construction of $\C$ yields Proposition~\ref{prop:lower}:

\begin{proposition}
\label{prop:lower}
	Let $u$ be a word of $B^*$ of length $k$. We have:
	\begin{eqnarray*}
		\prob{\C[\lambda,\theta]}(\widehat{u}\cdot\finish) & = & \theta^k\cdot\prob{\B_\lambda}(u)\\
		\prob{\C[\lambda,\theta]}((\widehat{u}\cdot\finish)^\ell) & = & (1-(1-\theta^k)^\ell)\cdot\prob{\B_\lambda}(u) \enspace.
	\end{eqnarray*}
\end{proposition}

It follows from Proposition~\ref{prop:lower} that the value of $\C[\lambda,\theta]$ is at least the value of $\B_\lambda$.

\begin{proposition}
\label{prop:theta}
	Let $u$ be a word of $(C\setminus\{\finish\})^+$. We have:$$\prob{\C[\lambda,\theta]}(u) \le \theta\ .$$
\end{proposition}

Proposition~\ref{prop:cheatonce} formalises the fact that there is no point in cheating after the first \finish\ letter:

\begin{proposition}
\label{prop:cheatonce}
	Let $u_1, \ldots, u_k$ be $k$ words of $(C\setminus\{\finish\})^*$ and $w$ be the word $u_1\cdot\finish\cdots u_k\cdot\finish$. Then, for any $1 \le i \le k$, if $u_i \notin \widehat{B^*}$, we have: $$\prob{\C[\lambda,\theta]}(w) \le \prob{\C[\lambda,\theta]}(u_i\cdot\finish\cdots u_k\cdot\finish)\ .$$
\end{proposition}
\begin{proof}
	After reading $u_1\cdots u_{i-1}\cdot\finish$, a run must be in one of the following three states: $q_0$, $q_0^\D$, and $\bot$. As $u_i \notin \widehat{B^*}$, reading it from $q_0^\D$ will lead to $\bot$. Thus, $$\prob{\C[\lambda,\theta]}(w) = \prob{\C[\lambda,\theta]}(q_0 \xrightarrow{ u_1\cdots u_{i-1}\cdot\finish}q_0) \cdot \prob{\C[\lambda,\theta]}(u_i\cdots u_k\cdot\finish)\ ,$$ and Proposition~\ref{prop:cheatonce} follows.
\end{proof}

Finally, Proposition~\ref{prop:upper} shows that $\C[\lambda,\theta]$ cannot have value 1 if $\B_\lambda$ does not.

\begin{proposition}
\label{prop:upper}
	For all words $w \in C^*$ such that $\prob{\C[\lambda,\theta]}(w) > \theta$, there exists a word $v \in B^*$ such that$$\prob{\B_\lambda}(u) \ge \frac{\prob{\C[\lambda,\theta]}(w) - \theta}{1-\theta}\ .$$
\end{proposition}

\begin{proof}
	Let us write $w=u_1\cdot\finish\cdots u_k\cdot\finish$ with $u_1, \ldots, u_k \in (C\setminus\{\finish\})^*$. By Proposition~\ref{prop:theta}, $k>1$, and by Proposition~\ref{prop:cheatonce} we can assume that $u_2,\ldots,u_k$ belong to $\widehat{B^*}$. Let $v_2,\ldots,v_k$ be the words of $B^*$ such that $u_i = \widehat{v_i}$. The $\C[\lambda,\theta]$-value of $w$ can be seen as a weighted average of 1 (the initial cheat, with a weight of $\theta$) and the $\B_\lambda$ values of the $v_i$'s (the weight of $\prob{\B_\lambda}(v_i)$ is the probability that the run enters $\D$ while reading $u_i$).  It follows that at least one of the $v_i$'s has a $\B_\lambda$-value greater than the $\C[\lambda,\theta]$-value of $w$
\end{proof}

Thus, for each $\lambda$ and $\theta$, the value of $\C[\lambda,\theta]$ is 1 if and only if the value of $\B_\lambda$ is 1. As all the $B_\lambda$'s have the same value, which is equal to the value of $\A$, we get:
		$$\begin{array}{rcl}
					\exists \Delta, \val{\M[\Delta]} = 1	&	\Longrightarrow	&	\exists \lambda, \val{\B_\lambda} = 1\\
										&	\Longrightarrow	&	\val{\A} = 1\\
										&	\Longrightarrow &	\forall \lambda, \val{\B_\lambda} = 1\\
										&	\Longrightarrow &	\forall \Delta, \val{\M[\Delta]} = 1\ .
		\end{array}$$
		
Theorem~\ref{thm:main} follows.

\section{Consequences}
\label{sec:consequences}
In this section, we show several consequences of the recursive inseparability results from Theorem~\ref{thm:main}
and of the construction from Lemma~\ref{lem:construction}.
The first is a series of undecidability results for variants of the value $1$ problem.
The second is about probabilistic B\"uchi automata with probable semantics, as introduced in~\cite{BBG12}.

\subsection{The noisy value $1$ problems}
\label{subsec:noise}
Observe that Theorem~\ref{thm:main} implies the following corollary:

\begin{corollary}
Both the universal and the existential value $1$ problems are undecidable.
\end{corollary}

We can go further.
Note that the universal and existential value $1$ problems quantify over all possible probabilistic transition functions.
Here we define two more realistic problems for NPA, where the quantification is restricted to probabilistic transition
functions that are $\varepsilon$-close to a given probabilistic transition function:

\begin{itemize}
	\item \textbf{The noisy existential value $1$ problem:} given a NPA $\A$, a probabilistic transition function $\Delta$
	and $\varepsilon > 0$, determine whether 
	there exists $\Delta'$ such that $|\Delta' - \Delta| \le \varepsilon$ and $\val{\A[\Delta']} = 1$.
	\item \textbf{The noisy universal value $1$ problem:} given a NPA $\A$, a probabilistic transition function $\Delta$
	and $\varepsilon > 0$, determine whether 
	for all $\Delta'$ such that $|\Delta' - \Delta| \le \varepsilon$, we have $\val{\A[\Delta']} = 1$.
\end{itemize}

It follows from Lemma~\ref{lem:construction} that both problems are undecidable:
\begin{corollary}
Both the noisy universal and the noisy existential value $1$ problems are undecidable.
\end{corollary}

Indeed, we argue that the construction from Lemma~\ref{lem:construction} implies a reduction
from either of these problems to the value $1$ problem for simple PA, hence the undecidability.
Let $\A$ be a simple PA, the construction yields a NPA $\M$ such that:
\[
\val{\A} = 1\ \Longleftrightarrow\ \forall \Delta, \val{\M[\Delta]} = 1\ \Longleftrightarrow\ \exists \Delta, \val{\M[\Delta]} = 1\ .
\]
It follows that for any probabilistic transition function $\Delta$ and any $\varepsilon > 0$, we have:
\[
\begin{array}{ccc}
\val{\A} = 1 & \Longleftrightarrow\ & 
\left(\forall \Delta', |\Delta' - \Delta| \le \varepsilon \Longrightarrow \val{\M[\Delta']} = 1\right)\\
             & \Longleftrightarrow\ & 
\left(\exists \Delta', |\Delta' - \Delta| \le \varepsilon \wedge \val{\M[\Delta']} = 1\right)\ .
\end{array}
\]
This completes the reduction.

\subsection{Probabilistic B\"uchi automata with probable semantics}
\label{subsec:buchi}

We consider PA over infinite words, as introduced in~\cite{BBG12}.
A probabilistic B\"uchi automaton (PBA) $\A$ can be equipped with the so-called probable semantics,
defining the language (over infinite words):
$$L^{>0}(\A) = \set{w \in A^\omega \mid \prob{\A}(w) > 0}\ .$$

It was observed in~\cite{BBG12} that the value $1$ problem for PA (over finite words)
easily reduces to the emptiness problem for PBA with probable semantics (over infinite words).

Informally, from a PA $\A$, construct a PBA $\A'$ by adding a transition 
from every final state to the initial state labelled with a new letter $\sharp$.
(From a non-final state, the letter $\sharp$ leads to a rejecting sink.)
As explained in~\cite{BBG12}, this simple construction ensures that $\A$ has value $1$ if and only if $\A'$ is non-empty, 
equipped with the probable semantics.

This simple reduction, together with Theorem~\ref{thm:main},
implies the following corollary:
\begin{corollary}
The two following problems for numberless PBA with probable semantics
are recursively inseparable:
\begin{itemize}
	\item all instances have a non-empty language,
	\item no instance has a non-empty language.
\end{itemize}
\end{corollary}


\section*{Acknowledgments}\label{sec:Acknowledgments}

We would like to thank the referees for their helpful comments.


%
%
%

\bibliographystyle{plain}
\bibliography{bib}

%
\end{document}